\documentclass[preprint,12pt]{elsarticle}

\usepackage{amssymb}
\usepackage{amsmath}
\usepackage{amsthm}

\usepackage{lineno}

\journal{Theoretical Computer Science}

\usepackage[T1]{fontenc}
\usepackage{graphicx}
\usepackage[textsize=footnotesize,color=green!40]{todonotes}
\usepackage{xspace}
\usepackage{complexity}
\usepackage{subcaption}
\usepackage[textsize=footnotesize,color=green!40]{todonotes}
\usepackage{thmtools}
\usepackage{hyperref}   
\usepackage{cleveref}
\usepackage{csquotes}

\usepackage{tikz}
\usetikzlibrary{shapes, calc, graphs, fit, positioning, backgrounds, arrows.meta}
\usepackage[most]{tcolorbox}

\newcommand{\gsfull}{\textsc{Geodetic Set}\xspace}

\newcommand{\TDM}{\textsc{3-Dimensional Matching}\xspace}
\newcommand{\fen}{\textsf{fen}\xspace}
\newcommand{\fvn}{\textsf{fvn}\xspace}
\newcommand{\mgs}{\textsf{mgs}\xspace}
\newcommand{\yes}{\texttt{yes}\xspace}

\newcommand{\problemdec}[3]{
	\begin{tcolorbox}[width = \textwidth,colback=white,arc=0pt,outer arc=0pt,boxrule=0.7pt,left =0.5em,right=0em]#1		\\[2pt]
		\begin{tabular}{ @{}l p{0.84\textwidth} c }
			\textbf{Input:} & #2 \\[.5pt]
			\textbf{Task:} & #3
		\end{tabular}
	\vspace{-0.25em}
	\end{tcolorbox}
}

\newtheorem{observation}{Observation}{\bfseries}{\itshape}
\newtheorem{myclaim}{Claim}{\bfseries}{\itshape}
\newtheorem{proposition}{Proposition}{\bfseries}{\itshape}
\newtheorem{definition}{Definition}{\bfseries}{\itshape}
\newtheorem{lemma}{Lemma}{\bfseries}{\itshape}

\crefname{observation}{observation}{observations}
\Crefname{observation}{Observation}{Observations}
\crefname{myclaim}{claim}{claims}
\Crefname{myclaim}{Claim}{Claims}
\crefname{proposition}{proposition}{propositions}
\Crefname{proposition}{Proposition}{Propositions}
\crefname{definition}{definition}{definitions}
\Crefname{definition}{Definition}{Definitions}
\crefname{lemma}{lemma}{lemmas}
\Crefname{lemma}{Lemma}{Lemmas}

\title{Algorithms and Hardness for Geodetic Set on Tree-like Digraphs} 

\author[1]{Florent Foucaud\fnref{fn1}}
\author[2]{Narges Ghareghani\fnref{fn2}}
\author[1]{Lucas Lorieau\fnref{fn3}}  
\author[2]{Morteza Mohammad-Noori}
\author[2]{Rasa Parvini Oskuei}
\author[3]{Prafullkumar Tale\fnref{fn4}}

\fntext[fn1]{This author was supported by the French government IDEX-ISITE initiative 16-IDEX-0001 (CAP 20-25), the International Research Center "Innovation Transportation and Production Systems" of the I-SITE CAP 20-25, and by the ANR project GRALMECO (ANR-21-CE48-0004).}
\fntext[fn2]{The work of this author is based upon research funded by Iran National Science Fundation (INSF) under project No. 4032600.}
\fntext[fn3]{This author has received financial support from the CNRS through the MITI interdisciplinary programs and the IRL ReLaX.}
\fntext[fn4]{This work is supported by INSPIRE Faculty Fellowship offered by DST, GoI, India.}

\affiliation[1]{organization={Université Clermont Auvergne, CNRS, Clermont Auvergne INP, Mines Saint-Étienne, LIMOS},
            city={Clermont-Ferrand},
            postcode={63000}, 
            country={France}}
\affiliation[2]{organization={School of Mathematics, Statistics, and Computer Science, University of Tehran},
            city={Tehran},
            country={Iran}}
\affiliation[3]{organization={Indian Institute of Science Education and Research Pune},
            city={Pune},
            country={India}}

\begin{document}

\begin{frontmatter}
\begin{abstract}
In the \gsfull problem, an input is a (di)graph
\(G\) and integer \(k\), and the objective 
is to decide whether there exists a vertex subset \(S\) of size \(k\)
such that any vertex in \(V(G)\setminus S\)
lies on a shortest (directed) path between two vertices in \(S\).
The problem has been studied on undirected and directed
graphs from both algorithmic and graph-theoretical
perspectives. 

We focus on directed graphs
and prove that \textsc{Geodetic Set} admits a 
polynomial-time algorithm on ditrees, that is, digraphs 
\emph{with} possible \(2\)-cycles when the 
underlying undirected graph is a tree (after deleting possible parallel edges). 
This positive result naturally leads us to investigate 
cases where the underlying undirected graph is `close to 
a tree'.  

Towards this, 
we show that \textsc{Geodetic Set} on digraphs 
\emph{without} \(2\)-cycles and
whose underlying undirected graph has feedback edge set 
number \(\fen\), can be solved in time  
\(2^{\mathcal{O}(\fen)} \cdot n^{\mathcal{O}(1)}\),  
where \(n\) is the number of vertices.  
To complement this, we prove that the problem remains 
\NP-hard on DAGs (which do not contain \(2\)-cycles) 
even when the underlying undirected graph 
has constant feedback vertex set number and  constant pathwidth.  
Our last result significantly strengthens the result of
Ara\'ujo and Arraes~[Discrete Applied Mathematics, 2022] 
that the problem is \NP-hard on DAGs 
when the underlying undirected graph is either 
bipartite, cobipartite or split.
\end{abstract}

\begin{keyword}
Geodetic Set  \sep Directed Trees \sep \NP-hardness \sep Parameterized Complexity \sep Feedback Edge Set Number \sep Feedback Vertex Set Number
\end{keyword}

\end{frontmatter}

\section{Introduction}
Harary, Loukakos, and Tsuros~\cite{harary1993} introduced the concept of a 
\emph{geodetic set}, defined as a set \(S\) of vertices of an undirected graph \(G\) 
such that every vertex of \(G\) lies on some geodesic (shortest path) between two 
vertices of \(S\). Since then, the problem of computing a minimum geodetic set has been 
extensively studied, both from the structural and algorithmic perspectives. 
It has become a central topic in \emph{geodesic convexity} in graphs~\cite{farber1986,bookGC}, 
and has found applications in diverse settings. We refer the reader to~\cite{ekim2012} 
and the references therein for a representative list of applications. 
For example, computing a minimum-size geodetic set can be seen as a network design problem, 
where one seeks to determine optimal locations of public transportation hubs in a 
road network~\cite{floCALDAM20}.

Let \gsfull be the corresponding algorithmic problem of determining a smallest possible geodetic set. 
Harary, Loukakos, and Tsuros~\cite{harary1993} proved that the problem is \NP-hard. 
See~\cite{JCMCC96} for the earliest rigorous proof. 
Later works established that the problem remains \NP-hard even for restricted graph 
classes~\cite{floISAAC20,floCALDAM20,DBLP:journals/tcs/ChakrabortyGR23,dourado2010,ekim2012}. 
This has motivated the study of algorithms for structured graph 
classes~\cite{wellpart,floISAAC20,DBLP:journals/tcs/ChakrabortyGR23,dourado2010,JCMCC96,ekim2012,mezzini2018}. 
To cope with this hardness, the problem has also been investigated through the lenses of 
parameterized complexity~\cite{floISAAC20,floICALP24,floSTACS25,KK22,T25} and approximation~\cite{floCALDAM20,DIT21}. 
More recently, interest in this problem has been rejuvenated, as its `metric-based nature' 
has led to interesting conditional lower bounds in parameterized 
complexity~\cite{floICALP24,DBLP:conf/stacs/FoucaudGK0IST25,T25} and enumeration complexity~\cite{BDM24}, 
together with related problems of a similar `metric-based' nature.
Here, we use the term `metric-based graph problem' as an umbrella notion 
for graph problems whose solutions are 
defined using a graph metric, for example, 
shortest distance between two vertices in the case of \gsfull.

While the majority of studies on the \gsfull problem 
focuses on undirected graphs, the problem has also been investigated for 
directed graphs (digraphs). 
Most of the existing work on geodetic sets in digraphs has concentrated on
non-algorithmic questions like determining the minimum and maximum values of these sets, 
as well as the range of possible values. 
See~\cite{chang2004geodetic,Chartrand2003,chartrand2000geodetic,dong2009upper,Farrugia05,lu2007geodetic} 
and the book~\cite[Chapter 6]{bookGC}. 
More recently, Ara\'ujo and Arraes~\cite{AraujoA22} initiated the algorithmic study 
of the \gsfull problem for digraphs. 
Before presenting their results, we formally define the problem addressed in this article 
and compare the results for undirected and directed graphs.

\problemdec{\gsfull}{A directed graph $D$ and an 
integer \(k\).}{Determine whether there exists a subset
\(S \subseteq V(D)\) of size \(k\) such that 
every vertex in \(V(D)\) lies on some directed shortest path
between two vertices in \(S\).}

We define a \emph{\(2\)-cycle} of $D$ as a directed cycle of length 
\(2\): a pair of vertices $u, v$ such that
both arcs \((u, v)\) and \((v, u)\) are present in $D$.
Note that directed graphs with \(2\)-cycles 
inherit the difficulties of designing 
algorithms for the undirected case (since the problem is equivalent on an undirected graph $G$ and on the digraph obtained from $G$ by replacing each edge by a directed 2-cycle).
Hence, the presence of \(2\)-cycles 
plays a critical role when stating results
for digraphs.

As noted by Ara\'ujo and Arraes~\cite[Sec.~6]{AraujoA22},
the unique minimum-sized 
geodetic set of an undirected tree \(T\) 
is equal to the set of its leaves.
Define an \emph{extremal vertex} of a digraph as a vertex that has either no incoming or no outgoing arcs, and denote by $\textrm{Ext}(D)$ the set of extremal vertices of $D$ (in an oriented graph, this set contains all leaves). An \emph{oriented graph} is a digraph that do not contains any 2-cycle. The authors remark that in oriented trees, \emph{i.e.} oriented graphs whose underlying undirected graph is a tree, the similar following property holds.

\begin{proposition}[\cite{AraujoA22}, Proposition 6.1]\label{prop-oriented-tree}
    Let $D$ be an oriented tree. Then, $\textrm{Ext}(D)$ is a minimum geodetic set of $D$.  
\end{proposition}

Using a non-trivial set of ideas and 
careful case analysis, the authors of~\cite{AraujoA22} 
generalized this algorithm
to digraphs without \(2\)-cycles whose
underlying undirected graph is a cactus.
(A graph is called a \emph{cactus} if each 
block is either an edge or a cycle,
and hence every tree is a cactus graph.)
We note that their case analysis holds only
when the input digraph does not contain
a \(2\)-cycle. A digraph whose underlying graph is a tree is called a {\it ditree}~\cite{DBLP:conf/wg/DaillyFH23,DBLP:journals/dam/FoucaudGS24}. Note that a ditree, contrary to oriented trees, might contain 2-cycles.

\paragraph{Our results}
As our first result, we present a 
polynomial-time
algorithm for ditrees.

\begin{restatable}{theorem}{lineartimeditree}
\label{thm-linear-time-ditree}
\textsc{Geodetic Set} on ditrees admits a linear-time algorithm.
\end{restatable}

This naturally leads us to investigate the problem 
for cases where the underlying undirected graph of the input digraph is `close to a tree'. 
We consider the following two definitions of `closeness to trees' for an undirected graph \(G\): 
the minimum number of edges (respectively, vertices) that must be deleted from \(G\)
to obtain a forest.
The minimum size of a set of such edges (respectively, vertices) is
called the \emph{feedback edge set number} (respectively, \emph{feedback vertex set number})
of the graph, and is denoted by \(\fen(G)\) and \(\fvn(G)\), respectively.

Kellerhals and Koana~\cite{KK22} proved that the \gsfull\ problem admits
an algorithm running in time \(2^{\mathcal{O}(\fen(G)^2)} \cdot n^{\mathcal{O}(1)}\)
when the input is an undirected graph.
In the authors' own words, \emph{``It turns out to be quite effortful to
obtain fixed-parameter tractability, requiring the design and analysis of polynomial-time
data reduction rules and branching before employing the main technical trick: Integer Linear
Programming (ILP) with a bounded number of variables.''}
Improving the running time of this algorithm has remained a challenging open problem.
Note that obtaining a fixed-parameter tractable algorithm parameterized by \(\fen\) for 
\gsfull\ when the input digraph is allowed to have \(2\)-cycles
inherently encodes the difficulties encountered by the authors of~\cite{KK22}.
As our next result, we show that a significantly faster algorithm (that does not need to use any ILP) can be obtained
when considering the case in which \(2\)-cycles
are not allowed.

\begin{restatable}{theorem}{fenalgo}
\label{thm-fenalgo}
\textsc{Geodetic Set} on digraphs without
$2$-cycles whose underlying undirected graph 
has feedback edge set number 
$\fen$, admits an algorithm running in time 
\( 2^{\mathcal{O}( \fen )} \cdot n^{\mathcal{O}(1)}\), 
where \(n\)
is the number of vertices in the input digraph.
\end{restatable}

Finally, we turn our attention to the feedback vertex set number
and show that a fixed-parameter tractable algorithm for this parameter is not possible, under classical hypothesis.
We first introduce another parameter, the \emph{pathwidth} of a graph, which is informally a measure of how much the graph is close to being a simple path. A more formal definition is given in Section~\ref{sec:prelims}.
Recently, Tale~\cite{T25} proved that \gsfull,
restricted to undirected graphs, remains \NP-hard even when the feedback vertex set number and the pathwidth of the input graph are constant. 
This implies that \gsfull, restricted to directed
graphs with possible \(2\)-cycles, remains \NP-hard
even when the underlying undirected graph has
constant feedback vertex set number and pathwidth. 
We prove that a similar result holds even
when \(2\)-cycles (which help channel the hardness 
from the undirected case into the directed case), are absent.
In fact, we prove the result not only for digraphs without 
\(2\)-cycles but also for directed acyclic graphs (DAGs),
which do not have directed cycles of any length.
This significantly strengthens the results of Ara\'ujo and Arraes~\cite{AraujoA22}, which state that 
\gsfull\ is \NP-hard on DAGs even when the underlying
undirected graph is bipartite, co-bipartite, or a split graph.

\begin{restatable}{theorem}{fvnhardness}
\label{thm:fvn-hardness}
\textsc{Geodetic Set} on DAGs (which do not have \(2\)-cycles) 
whose underlying undirected graph 
has feedback vertex set number \(12\) and pathwidth $14$, is \NP-hard.
\end{restatable}

A preliminary version of this paper appeared in the proceedings of the conference CALDAM 2026~\cite{DBLP:conf/caldam/FoucaudGLNOT26}. The current paper contains all full proofs of our results, as well as a corrected proof of Theorem~\ref{thm-fenalgo} (in particular, the addition of Claim~\ref{claim-path-2extremals}).

\paragraph{Outline} 
We start with a preliminary section in Section~\ref{sec:prelims}. Then we present our linear-time algorithm on ditrees in Section~\ref{sec:trees} proving Theorem~\ref{thm-linear-time-ditree},
followed by describing the fixed-parameter tractable algorithm mentioned in Theorem~\ref{thm-fenalgo} in Section~\ref{sec:fes-fpt-algo}.
We prove \Cref{thm:fvn-hardness} in Section~\ref{sec:hard} and conclude with some open problems in Section~\ref{sec:conclu}.

\section{Preliminaries}
\label{sec:prelims}

We present general definitions and results that will be used throughout the paper. We refer to the book~\cite{DBLP:books/sp/CyganFKLMPPS15} for terminology and details on parameterized complexity.

    A {\it directed graph} (digraph for short) $D$ consists of vertex set $V(D)$ and arc set $A(D)$, where each arc is an ordered pair of vertices. An arc from vertex $v$ to vertex $u$ is denoted by $vu$, and $v$ is its \emph{tail} and $u$ is its \emph{head}.
	A digraph is called an {\it oriented graph} if it does not contain a directed 2-cycle.
	The \emph{underlying undirected graph} (or simply \emph{underlying graph}) of some digraph $D$ is the graph obtained by removing the orientation of each arc of $D$.
	An \emph{oriented path} of a digraph $D$ is a subgraph of $D$ whose underlying graph is a path. A \emph{directed path} (or \emph{dipath}) is an oriented path for which all arcs are oriented in the same direction.
	A digraph is called {\it strongly connected} if every pair of vertices are connected by a directed path. The \emph{in-neighborhood} of vertex $u$ is denoted by $N^-(u)$ and its {\it out-neighborhood} is denoted by $N^+(u)$. The in-neighborhood of a subset $S$ of $V(D)$ is $N^-(S)=(\cup_{u\in S}N^-(u))\setminus S$, and similarly the out-neighborhood of $S$ is defined. A vertex $x$ is called a {\it source}, if $N^-(x)=\emptyset$, and it is called a {\it sink} if $N^+(x)=\emptyset$. A vertex that is a source or a sink is called \emph{extremal}. 
	For two vertices $u$ and $v$, the set of vertices that lie in some shortest path from $u$ to $v$ is denoted by $I(u,v)$ and for a subset $S$ of $V$, the {\it geodetic closure} of $S$, denoted by $I(S)$, is the set of all vertices which lie in some shortest path between two vertices of $S$. In other words, $I(S)=\cup_{u,v \in S} \big(I(u,v)\cup I(v,u)\big)$. We also say that a vertex $v$ \emph{is covered} by two vertices $u$ and $w$ if $v\in I(\{u,w\})$.
	In a directed path $P$ from $u$ to $v$, vertices $u$ and $v$ are called the \emph{tail} and the \emph{head} of $P$, respectively. The vertices of $P$ which are neither the tail nor the head of $P$ are called its \emph{inner vertices}. A {\it directed acyclic graph} (DAG for short) is a digraph which does not contain any directed cycle. 

A vertex $v$ is \emph{transitive} if, for every in-neighbor $u_1$ and out-neighbor $u_2$ of $v$, 
either the arc $u_1u_2$ exists, or $u_1=u_2$.\footnote{Note that the latter condition was not present in the definition from~\cite{AraujoA22}, since the authors only considered digraphs without 2-cycles.} A vertex of a digraph is called a \emph{leaf} if, in the underlying undirected graph, it has degree~1. Note that any leaf of a digraph is either a sink, a source, or transitive.

The following lemma states a simple observation on madatory vertices for any geodetic set.
\begin{lemma}[\cite{AraujoA22}]\label{lem:sink-sources}
	   In any digraph $D$, every source, sink and transitive vertex of $D$ (in particular, every leaf of $D$) belongs to every geodetic set of $D$.
\end{lemma}

We conclude this section with some formal definition of the parameter pathwidth in an undirected graph. We call a \emph{path decomposition} of a graph $G=(V,E)$ a sequence $(B_1, \dots, B_\ell)$ of subsets of vertices of $V$ called \emph{bags} for which the following properties hold:
\begin{itemize}
    \item Every vertex of $V$ appears in at least one bag;
    \item For each edge of $E$, there exists at least one bag containing both of its endpoints;
    \item For any vertex $v\in V$, if $v\in B_i$ and $v\in B_j$, with $i<j$, then $v\in B_k$ for any $k\in [i,j]$.
\end{itemize}

For a path decomposition $\mathcal D = (B_1, \dots, B_\ell)$, we call the \emph{width} of $\mathcal D$ the value $\max_{1\leq i\leq \ell}|B_i|-1$. The \emph{pathwidth} of a graph $G$ is the minimum width among all possible path decompositions of $G$.
\section{Linear-time algorithm for ditrees}\label{sec:trees}
	
	In this section, we consider the problem of finding the geodetic number of a ditree
	and prove the following theorem.
	
	\lineartimeditree*
	
	In the context of digraphs admitting 2-cycles, we say that a vertex of a digraph $D$ is a \emph{leaf} of $D$ if it is a leaf of the underlying graph of $D$. When $T$ is an oriented tree, a minimum-size geodetic set (\mgs for sort)  may contain some non-leaf vertices: as observed in \Cref{prop-oriented-tree}, in this case, an optimal \mgs always consists of all sources and sinks of the tree.
	To prove \Cref{thm-linear-time-ditree}, we reduce the problem to finding a \mgs in directed trees where the only 2-cycles present are adjacent to some leaf of the graph. We argue then that taking all extremal vertices of the graph, in addition to the leaves contained in a 2-cycle, yields a \mgs. Intuitively, the obtained graph behaves almost exactly as an oriented tree, which enables us to naturally extend \Cref{prop-oriented-tree}.
	
	Let $S$ be a maximal strongly connected component of $D$. Then we call $S$ a {\it source set} if $N^-(S)\setminus S=\emptyset$. Similarly, we call $S$ a {\it sink set} if $N^+(S)\setminus S=\emptyset$. We can state a simple observation about sink and source sets.
		
\begin{observation}
    Let $S$ be a source set or a sink set. Then, any two adjacent vertices of $S$ form a 2-cycle.
\end{observation}

    We now define a new digraph, whose structure is very similar to an oriented tree.

\begin{definition} \label{def-contracted-tree}
		Let $T$ be a ditree. We define the \emph{contracted ditree} $T^c$ of $T$ by iteratively contracting every 2-cycle of $T$ that does not contain any leaf vertex into a single vertex.
	\end{definition}

In particular, in $T^c$, the only 2-cycles are adjacent to some leaf of the graph. A ditree $D$ with a sink set and its contracted ditree $D^c$ are shown in Figure~\ref{fig:ditree-contracted}. \Cref{lemma-diatree-sink} shows that in contrast to oriented trees, an \mgs in a general ditree may contain vertices which are neither sources, sinks nor leaves. Then, \Cref{lem-sink-sources-2} states that if all source and sink sets are either composed of a unique vertex or contain a leaf, then we can easily find a \mgs in the considered ditree. In particular, this can be applied on $T^c$. Finally, \Cref{lem-replace-in-sourceset,lemma-asli} justify that we can extend any \mgs of $T^c$ to $T$.

    	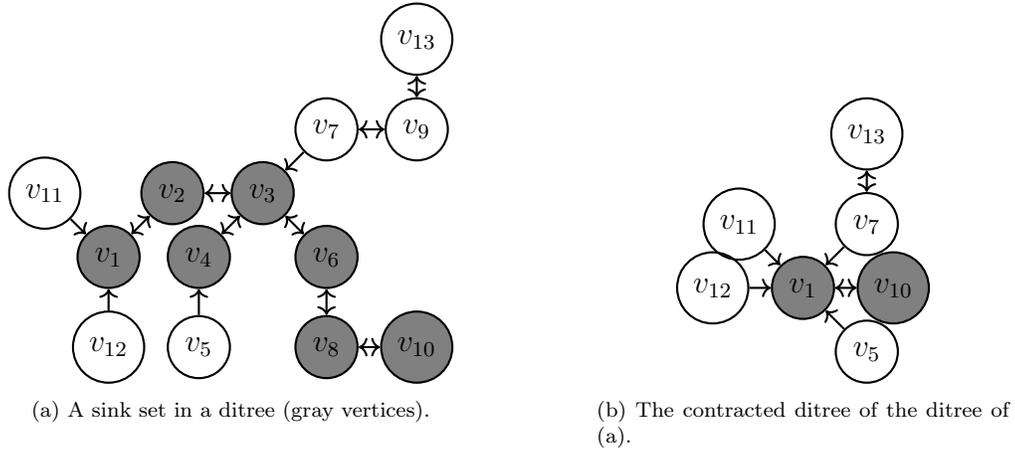
\begin{figure*}[t!]
    \begin{subfigure}[t]{0.49\textwidth}
        \centering
        \begin{tikzpicture}[node distance={12mm}, thick, main/.style = {draw, circle}]
    		\node[main, fill=gray] (1) {$v_1$};
    		\node[main, fill=gray] (2) [above right of=1] {$v_2$};
    		\node[main, fill=gray] (3) [right of=2] {$v_3$};
    		\node[main, fill=gray] (4) [below left of=3] {$v_4$};
    		\node[main] (5) [below of=4] {$v_5$};
    		\node[main, fill=gray] (6) [below right of=3] {$v_6$};
    		\node[main] (7) [above right of=3] {$v_7$};
    		\node[main, fill=gray] (8) [below of=6] {$v_8$};
    		\node[main] (9) [right of=7] {$v_9$};
    		\node[main, fill=gray] (10) [right of=8] {$v_{10}$};
    		\node[main] (11) [above left of=1] {$v_{11}$};
    		\node[main] (12) [below of=1] {$v_{12}$};
    		\node[main] (13) [above of=9] {$v_{13}$};
    
    		\draw [<->](1) -- (2);
    		\draw [<->](2) -- (3);
    		\draw [<->](3) -- (4);
    		\draw[<-] (4) -- (5);
    		\draw [<->](3) -- (6);
    		\draw[<-] (3) -- (7);
    		\draw [<->](6) -- (8);
    		\draw[<->] (7) -- (9);
    		\draw[<->] (8) -- (10);
    		\draw[<-] (1) -- (11);
    		\draw[<-] (1) -- (12);
    		\draw[<->] (9) -- (13);
		\end{tikzpicture}
		\caption{A sink set in a ditree (gray vertices).}
		\label{fig-sink-set}
    \end{subfigure}%
    \hfill
    \begin{subfigure}[t]{0.4\textwidth}
        \centering
        	\begin{tikzpicture}[node distance={12mm}, thick, main/.style = {draw, circle}, scale=.6]
				\node[main, fill=gray] (1) {$v_1$};
				\node[main] (5) [below right of=1] {$v_5$};
				\node[main] (7) [above right of=1] {$v_7$};
				\node[main,fill=gray] (10) [right of=1] {$v_{10}$};
				\node[main] (11) [above left of=1] {$v_{11}$};
				\node[main] (12) [left of=1] {$v_{12}$};
				\node[main] (13) [above of=7] {$v_{13}$};

				\draw[<-] (1) -- (5);
				\draw[<-] (1) -- (7);
				\draw[<->] (1) -- (10);
				\draw[<-] (1) -- (11);
				\draw[<-] (1) -- (12);
				\draw[<->] (7) -- (13);
			\end{tikzpicture}
			\caption{The contracted ditree of the ditree of (a).}
			\label{fig:contracted -tree}
    \end{subfigure}
    \caption{An example of a ditree $T$ and its contracted ditree $T^c$.}
    \label{fig:ditree-contracted}
\end{figure*}

	\begin{lemma}
	    \label{lemma-diatree-sink}
	    Let $T$ be a ditree and let $S$ be a source set or sink set of $T$. Then, every geodetic set of $T$ includes at least one vertex from $S$.
	\end{lemma}
	
	\begin{proof}
	    Since $S$ is a source or sink set, no directed path with both endpoints outside $S$ can exist. Thus, in a geodetic set of $T$, no two vertices outside $S$ can cover any vertex of $S$. Hence, every geodetic set of $T$ contains at least one vertex from $S$.
	\end{proof}

\begin{lemma}\label{lem-sink-sources-2}
    Let $T$ be a ditree for which any source (resp. sink) set of size at least $2$ contains a leaf. Let $S\subseteq V(T)$ be the set of all sink (resp. source) vertices and all leaves of $T$. Then, $S$ is a geodetic set of $T$ of minimum size.
\end{lemma}

\begin{proof}
	    Using Lemma \ref{lem:sink-sources}, we know that every geodetic set of $T$ contain $S$. So, if we show that  $S$ is a geodetic set, it would have minimum size. To prove that, let $w\in V(T)\setminus S$. We show that there exists a pair $(v,u)$ of vertices of $S$ such that $w$ belongs to the shortest path connecting $v$ to $u$. Since $w$ is not a sink, there exists a vertex $u_1$ such that $wu_1\in A(T)$. If $u_1\in S$, we define $u=u_1$, otherwise $u_1$ is neither a sink nor a leaf. Hence, there exists a vertex $u_2$ different from $w$ such that $u_1u_2\in A(T)$. This process must stop at some stage $i$, which means that $u_i\in S$ and there is a directed path from $w$ to $u_i$ (we let $u=u_1$). We can prove similarly that there exists a vertex $v\in S$ which has a directed path to $w$. Therefore, there is a directed path in $T$ from $v$ to $u$ passing through $w$. Since the underlying graph of $T$ is a tree there is a unique directed path from $v$ to $u$, hence, $w$ is covered by vertices $v$ and $u$, as desired. 
	\end{proof}

	\begin{lemma}\label{lem-replace-in-sourceset}
	    Let $T$ be a ditree and $S$ be a source set or sink set of size at least~$2$ that contains no leaves. Suppose $M$ is an \mgs for $T$ that includes exactly one vertex from $S$. Define $N\subseteq V(T)$ as the set obtained from $M$ by replacing the unique vertex of $S$ with any other vertex of $S$. Then, $N$ is also an \mgs for $T$. 
	\end{lemma}
	\begin{proof}
	Without loss of generality, assume that $S$ is a source set. Arguments are similar for sink sets. 
	Let $u\in S\cap M$, we first show that the lemma holds when $u$ is replaced by one of its neighbours in $S$. Let $P_1, \ldots, P_m$ be all maximal directed paths starting from $u$ and ending at some vertex of $M$. Choose a vertex $v\in S$ such that $v$ is an out-neighbour of $u$. Since $S$ is a source set, $v$ is also an in-neighbour of $u$. Because $M$ is an \mgs and $S$ is a source set, some paths among $P_1, \ldots, P_m$ must contain $v$. let $P_1, \ldots, P_i$ denote exactly those paths that contain $v$. Since $T$ is a ditree, each such path necessarily begins with the arc $uv$. 
	    Moreover, as $u$ is not a leaf, there should also exist paths among
	    $P_1, \ldots, P_m$ that do not contain $v$. Assume without loss of generality, that $i<m$ and $P_{i+1}, \ldots, P_m$ are all these paths. 
	    
	    Now, for each $1\leq j\leq i$, let $P'_j$ be the path obtained from $P_j$ by deleting the vertex $u$ and the arc $uv$. For each $j$ with $i+1\leq j\leq m$, let $P'_j$ be the directed path obtained from $P_j$ by prepending the arc $uv$. 
	    It follows that the set $N$ covers exactly the same set of vertices as $M$, and hence $N$ is also an \mgs. 
	    
	  Next, we show that the lemma remains valid if $u$ is replaced by some vertex $w\in S$ that is not a neighbour of $u$. Since $S$ is a source set, there is a directed path from $u$ to $w$ all whose  internal vertices are in $S$. Moreover, each arc in this path is part of a 2-cycle in $T$. Therefore, starting from $u$, we may successively replace the current vertex by its neighbour along the path, until reaching $w$. At each step, the resulting set is an \mgs and thus, after the final replacement, the set obtained by substituting $u$ with $w$ is also an \mgs, as required.
	\end{proof}
	
	\begin{lemma}\label{lemma-asli}
	    Let $T$ be a ditree and $ \mathcal{S}=\{S_1, S_2, \ldots, S_t\}$ be the set of all source sets and sink sets of $T$ which do not contain any leaf. Let $\mathcal{L}$ be the set consisting of all leaves of $T$ and $M$ be a subset of $V(T)$ that contains $\mathcal{L}$ and exactly one vertex from each $S_i$. Then, $M$ is an \mgs for $T$.
	\end{lemma}
	\begin{proof}
	    We prove the theorem
	    using induction on  the number of vertices of $T$. If  $T$ has $2$ vertices, then every vertex of $T$ is a leaf. Using Lemma \ref{lem:sink-sources}, $\mathcal{L}$ is an \mgs for $T$. By the induction hypothesis, suppose that
	   the theorem is true for any ditree of size less than $n$. Let $T$ be a ditree with $n$ vertices. If each source (resp. sink) set of $T$ of size at least $2$ contains a leaf, then each source (resp. sink) set which does not contain a leaf is a source (resp. sink) vertex. So in this case, the theorem holds, using Lemma \ref{lem-sink-sources-2}.
	   
	   Now, without loss of generality, suppose that $S_t \in \mathcal{S}$ is a source set and it has at least two vertices. Let $u,v$ be two vertices in $S_t$ which are adjacent and $T'$ be a ditree obtained from $T$ by contracting the edge between $u$ and $v$. Let $w$ be the new vertex in $T'$ that replaces two vertices $u$ and $v$ of $T$. Note that $T'$ is the ditree with the same set of leaves as $T$, and with the same set of source (resp. sink) sets as $T$, except that $S_t$ is replaced by $S'_t$, where $S'_t=(S_t\setminus \{u,v\}) \cup \{w\}$. Hence, by the induction hypothesis, there is an \mgs of $T'$ that contains all leaves and a vertex from each of the set $S_1, S_2, \ldots, S'_t$. 
	   
	   By Lemma \ref{lem-replace-in-sourceset}, we conclude that a subset of vertices of $T'$ that contains all leaves of $T'$ and an arbitrary vertex from each of the set $S_1, S_2, \ldots, S'_t$ is an \mgs. 
	   
	    Let $P'_1, \ldots, P'_r$ be the set of all maximal directed paths in $T'$ starting from $w$. Since $w$ is not a leaf, $r\geq 2$. One can see that the end-vertex of each of these paths is either a leaf or belongs to a sink set. Moreover, no two of these paths terminate at the same sink set; Otherwise the underlying undirected graph would contain a cycle, which is impossible.
	    Since $T'$ has an \mgs which contains all leaves and exactly one vertex from each source (resp. sink) set which does not contain any leaf, using Lemma~\ref{lem-replace-in-sourceset}, $T'$ has an \mgs, $M'$, which contains $w$ and the end-vertices of each path $P'_1, \ldots, P'_r$.
	    
	   Define $M=M'\setminus \{w\}\cup \{u\}$. We claim that $M$ is an \mgs for $T$. Since $T$ and $T'$ have the same set of leaves and the same number of sink (resp. source) sets, Lemma~\ref{lem-sink-sources-2} implies that every \mgs of $T$ contains at least $|M'|$ vertices. Therefore, it remains to show that $M$ is a geodetic set for $T$.
	   
	    Since $u$ and $v$ are in a same source set, $T$ contains both arcs $uv$ and $vu$. Since $M'$ is an \mgs for $T'$, each vertex of $T$ other than $v$ is covered by a path with both ends in $M$. Now, it suffices to show that $v$ is covered by some path with ends in $M$ (note that $T$ is a ditree, so each path is a shortest path between its end vertices), Now, we claim that $v$ belongs to some $P_i$, $1\leq i\leq r$. If not, $T\setminus \{v\}$ is connected and this happens only if $v$ is a leaf, which is not the case, a contradiction.  
	    \end{proof}

\begin{proof}[Proof of Theorem~\ref{thm-linear-time-ditree}]
By Lemma~\ref{lemma-asli}, every \mgs of $T$ contains all leaves and an arbitrary vertex of any source (sink) set that does not contain any leaf. Since any leaf in $T$ is also a leaf in $T^c$ and vice-versa, the algorithm proceeds as follows. First, we construct the contracted ditree $T^c$ from $T$ in linear time by contracting every 2-cycle that is not incident with a leaf. Then, we determine all source (resp. sink) vertices (and the corresponding source (sink) sets in $T$), and leaves of $T^c$, which can also be done in linear time by examining the in-degree and out-degree of each vertex. Since each step requires only linear time in the number of vertices and edges, the overall complexity is $O(|V(T)|+|E(T)|) = O(n)$. 
\end{proof}

\section{Algorithm parameterized by feedback edge set number}
\label{sec:fes-fpt-algo}
    
    In this section, we present an algorithm solving \gsfull parameterized by the feedback edge set number of the input graph and prove \Cref{thm-fenalgo}. 
    
    \fenalgo*

    We introduce useful notions from \cite{foucaudMonitoring2023}. 
    A \emph{core vertex} is a vertex of degree at least 3. 
    A \emph{core path} of digraph $D$ is a path in the underlying graph of $D$ between two core vertices with 
    only degree~2 internal vertices. 
    Note that both endpoints are allowed to be the same core vertex: in this case, we call the core path a \emph{core cycle}.
    We call \emph{proper core path} a core path whose endpoints are two different core vertices. 
    A \emph{leg} of the underlying graph of $D$ is a (non-empty) path between a core vertex and a leaf in said graph. 
    The \emph{base graph} of some undirected graph $G$ is the graph obtained by iteratively removing leaves from $G$ until no leaf is present. We say that the base graph of a digraph $D$ is the base graph of its underlying undirected graph. Note that a  vertex of a base graph is either a core vertex, or an inner vertex of some core path. The following observation comes from \cite[Observation 5]{KK22}.
    \begin{observation}\label{obs-fen-bounds}
        The base graph of any undirected graph $G$ has at most $2\fen(G) - 2$ core vertices and at most $3\fen(G) -3$ core paths.
    \end{observation}
    
    We say that the \emph{base digraph} $D_b$ of $D$ is the subgraph of $D$ such that its underlying graph is the base graph of the underlying graph $D$. 
    A \emph{hanging tree} of the underlying graph of $D$ is the union of some legs iteratively removed to form the base graph of $D$ so that the union of those legs forms a connected component. The \emph{root} of a hanging tree of $D$ is the vertex of the base graph that was linked to the last removed leg of the considered hanging tree. It is easily seen that the underlying undirected graph of $D$ can be decomposed into its base graph and a set of maximal hanging trees. 
    A \emph{hanging ditree} of $D$ is a subdigraph of $D$ such that its underlying graph is a hanging tree of $D$, and its root is the root of the associated hanging tree. Similarly, $D$ can be decomposed into its base digraph and a collection of maximal hanging ditrees. We call an \emph{oriented core path} any core path of the base graph of $D$ whose edges are oriented based on the orientations in $D$. If an oriented core path forms a directed path in $D$, we call it a \emph{core dipath}. \Cref{fig-defs} illustrates some of these definitions by an example digraph and its core digraph.
    
    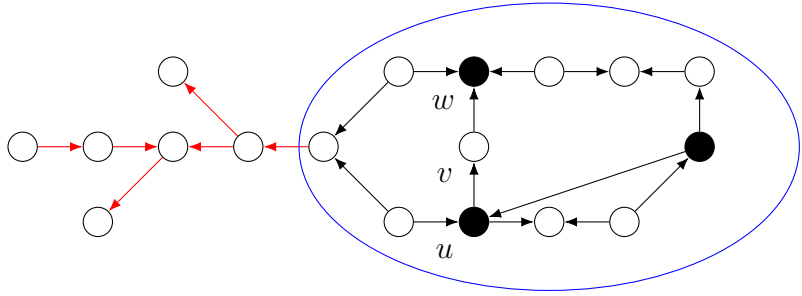
\begin{figure}
        \centering
        \begin{tikzpicture}[every node/.style={draw, circle}, >=Latex]
            \node (A) at (0,0) {};
            \node (B) at (1,0) {};
            \node (C) at (1,-1) {};
            \node (D) at (2,0) {};
            \node (E) at (2, 1) {};
            \node (F) at (3,0) {};
            \node (G) at (4,0) {};
            
            \begin{scope}[every node/.append style={fill=black}]
            \node[label=below left:$w$] (J) at (6, 1) {};
            \node[label=below left:$u$] (L) at (6, -1) {};
            \node (S) at (9, 0) {};
            \end{scope}
            
            \node (H) at (5,1) {};
            \node (I) at (5, -1) {};
            \node[label=below left:$v$] (K) at (6, 0) {};
            \node (M) at (7, 1) {};
            \node (N) at (7, -1) {};
            \node (P) at (8, 1) {};
            \node (Q) at (8, -1) {};
            \node (R) at (9, 1) {};
            
            \begin{scope}[red]
            \draw[->] (A) -- (B);
            \draw[->] (B) -- (D);
            \draw[->] (D) -- (C);
            \draw[->] (F) -- (D);
            \draw[->] (F) -- (E);
            \draw[->] (G) -- (F);
            \end{scope}
            
            \draw[->] (H) -- (G);
            \draw[->] (I) -- (G);
            \draw[->] (H) -- (J);
            \draw[->] (I) -- (L);
            \draw[->] (K) -- (J);
            \draw[->] (L) -- (K);
            \draw[->] (L) -- (N);
            \draw[->] (M) -- (J);
            \draw[->] (Q) -- (N);
            \draw[->] (Q) -- (S);
            \draw[->] (M) -- (P);
            \draw[->] (R) -- (P);
            \draw[->] (S) -- (R);
            \draw[->] (S) -- (L);
            
            \node[draw, ellipse, inner sep=1.5mm, blue, fit= (H) (I) (J) (K) (L) (M) (N) (P) (Q) (R) (S)] {};
            
        \end{tikzpicture}
        \caption{In this example digraph $D$, the red arcs form a hanging ditree. The circled blue subgraph is the core digraph of $D$. It is composed of three core vertices in black, and five core oriented paths. In particular, the path $uvw$ is a core dipath.}
        \label{fig-defs}
    \end{figure}
    
    First, we will argue that deciding which vertices to take in a solution in the hanging ditrees of $D$ is not difficult, using the following observation.
    
    \begin{observation}\label{obs-hanging}
        Let $v$ be a vertex of a hanging ditree $T$ of $D$ rooted in $r$.
        \begin{itemize}
            \item If $v$ has an outgoing arc, $v$ can either reach $r$ or a sink $w\in V(T)$.
            \item If $v$ has an incoming arc, $v$ can be reached by either $r$ or a source $u\in V(T)$.
        \end{itemize}
    \end{observation}
    
     \begin{proof}
        We prove only the first case, as the other case is very similar. Suppose $v$ cannot reach the root $r$ but has some outgoing neighbor in $T$. Consider $w$ a vertex reachable by $v$ with maximum distance to $v$. Vertex $w$ is well-defined since the vertices reachable from $v$ are in $T$ and $T$ has no directed cycle. By maximality of the distance to $v$ in $T$, $w$ has to be a sink.
    \end{proof}

    \begin{myclaim}\label{claim-trees-covered}
        Let $D$ be a digraph with $S^0$ its set of extremal vertices. Suppose $S$ is a geodetic set of $D$. Then, $S'=(S\cap V(D_b)) \cup S^0$ is a geodetic set of $D$.
    \end{myclaim}
    
    \begin{proof}
        Let $v$ be a vertex of $V(D)\setminus S$. The set $S$ is a geodetic set of $D$, so $v$ is located on a shortest path from $u_0$ to $w_0$, with $u_0,w_0\in S$. If both $u_0$ and $w_0$ are either in $D_b$ or some extremal vertex of a hanging ditree, we have $v\in I(S')$. Suppose then it is not the case, and suppose that $u_0$ belongs to some hanging ditree $T$ rooted in $r$ and is not an extremal vertex. (A similar proof can be done when $w_0$ belongs to some hanging ditree and is not an extremal vertex).
        \Cref{obs-hanging} states that $u_0$ can be reached by either a source $u$ in $T$ or by $r$. We consider both cases, and describe for each a shortest path with endpoints in $S'$ covering $v$.
        First, suppose that $u_0$ can be reached by some source $u\in T$, then $v\in I(u,w_0)$ and we are done.
        Otherwise, $u_0$ can be reached by $r$, hence the vertices reached by $u_0$ are vertices of $T$, so $v$ belongs to $T$. By applying again \Cref{obs-hanging} with $v$, there exists a sink $w$ reachable by $v$ in $T$, and $w\in S'$.
        
        It remains to find a source in $S'$ that can reach $r$. If $r$ is in $S$, then $v\in I(r,w)$ and we are done. Otherwise, $r$ is covered by some shortest path with a starting endpoint $u'$ in $S$. If $u'\in S'$, then a shortest path from $u'$ to $r$ and the unique path from $r$ to $w$ going through $v$ form together a shortest path covering $v$ and we are done. Finally, suppose $u'$ is not an extremal vertex and belongs to some hanging ditree $T'$ rooted in $r'$. If $T=T'$, by applying \Cref{obs-hanging}, there exists a source $u''$ that can reach $u'$ in $T$, and the unique path from $u''$ to $w$ covers $v$. Otherwise, $u'$ must reach $r'$ (or it cannot reach $r$) and similarly, we can find a source $u''$ in $T'$ that covers $v$ with $w$.
    \end{proof}
    
    \Cref{claim-trees-covered} identifies which vertices of hanging ditrees are part of minimum geodetic sets, so we next focus on vertices of oriented core paths. We identify four different cases, depending on the number of extremal vertices present in the considered oriented core path. Note that when there is no extremal vertex among the inner vertices of an oriented core path, it is in fact a core dipath.
    
    \begin{myclaim}\label{claim-path-3extremals}
        Let $D$ be a digraph and $S^0$ the set of extremal vertices of $D$. Consider an oriented core path $P$ of $D$ and denote by $V_P$ the set of its inner vertices. Number them as $V_P=\{v_1, \dots, v_l\}$ so that two neighbors in $P$ have consecutive indexes. Suppose that $|S^0\cap V_P|\geq 3$, and denote by $v_i$ (respectively $v_j$) the vertex with minimum index (respectively with maximum index) of $S^0\cap V_P$. We have $\{v_i, \dots, v_j\}\subseteq I(S^0)$ and there exists an \mgs $S$ of $D$ so that $V_P\cap S = V_P\cap S^0$.
    \end{myclaim}
	
	\begin{proof}
	    We argue that for any $k\in[i,j]$, $v_k$ is covered by two vertices of $S^0$. Indeed, since $|S^0\cap V_P|\geq 3$, there exist other extremal vertices between $v_i$ and $v_j$. Consider $v_m$ the vertex with minimum index among those vertices. If $v_i$ is a sink, then $v_m$ is a source (and vice versa), so if $k\in[i,m]$, $v_k\in I(v_i, v_m)$, since there exist a unique path between $v_i$ and $v_m$. One can apply iteratively this argument with $v_m$ and the other vertex of $S^0\cap V_P$ it can reach (or that can reach $v_m$, respectively).
	    
	    Suppose now that $S$ is an \mgs of $D$, and that some vertex $v_p$ belongs to $(V_P\cap S)\setminus S^0$. Since $\{v_i, \dots, v_j\}\subseteq I(S^0)$, if $i<p<j$, $S\setminus\{v_p\}$ is still a geodetic set, which is in contradiction with the minimality of $S$. Suppose then without loss of generality that $p < i$. Since $S$ is a geodetic set, the core vertex $v^\leftarrow$ neighboring $v_1$ is either in $S$ or covered by a shortest path starting at some vertex $u\in S$ and ending at some vertex $w\in S$. In the former case, vertices $v_1, \dots, v_i$ are covered by the unique shortest path between $v^\leftarrow$ and $v_i$ and $v_p$ can be removed from $S$ as above. In the latter case, consider $S' = (S\setminus \{v_p\})\cup\{v^\leftarrow\}$. Any shortest path for which $v_p$ is an endpoint can either be extended to a shortest path with $v^\leftarrow$ as an endpoint, or goes through $v^\leftarrow$. Thus, $S'$ is also an \mgs of $D$. The same reasoning applies if $j<p$ by considering $v^\rightarrow$, the core vertex neighboring $v_l$.
	\end{proof}
	
	The proof for the case where there are only two extremal vertices in the oriented core path is very similar to the previous case.
	
	\begin{myclaim}\label{claim-path-2extremals}
	    Let $D$ be a digraph and $S^0$ the set of extremal vertices of $D$. Consider an oriented core path $P$ of $D$ and denote by $V_P$ the set of its inner vertices. Number them as $V_P=\{v_1, \dots, v_l\}$. Suppose that $|S^0\cap V_P|= 2$, and denote by $v_i$ (respectively $v_j$) the vertex with minimum index (respectively with maximum index) of $S^0\cap V_P$. There exists an \mgs $S$ of $D$ so that $V_P\cap S \subseteq \{v_i, v_{i+1}, v_j\}$. In particular, $V_P\cap S$ contains at most one non-extremal vertex.
	\end{myclaim}
	
	\begin{proof}
	    We consider two different cases, depending on whether the dipath between $v_i$ and $v_j$ is a shortest path or not in $D$. If it is the case, then the vertices of $\{v_{i+1}, \dots, v_{j-1}\}$ are covered by $v_i$ and $v_j$. If this path is not a shortest path between $v_i$ and $v_j$ any \mgs of $D$ must contain one vertex $v_k$ for $k\in (i,j)$. Indeed, the only accessible extremal vertices for vertex $v_k$ are $v_i$ and $v_j$. Furthermore, any vertex $v_k$ is suited for covering vertices in $\{v_{i+1}, \dots, v_{j-1}\}$; in particular $v_i+1$ is. The other vertices of $P$ can be handled by the endpoints of the core path, as in the proof of \Cref{claim-path-3extremals}.
	\end{proof}
	
	\begin{myclaim}\label{claim-path-1extremal}
	    Let $P$ be an oriented core path of some digraph $D$ and denote by $V_P$ the inner vertices of $P$. Number vertices of $P$ such that $V_P = \{v_1, \dots, v_l\}$. Suppose that $|V_P\cap S^0| = 1$ and denote the vertex of $V_P\cap S^0$ by $v_p$. There exists an \mgs $S$ of $D$ such that $|V_P\cap S|\leq 3$ and $V_P\cap S \subseteq \{v_1, v_p, v_l\}$. In particular, $V_P\cap S$ contains at most two non-extremal vertices.
	\end{myclaim}
	\begin{proof}
    Consider an \mgs $S'$ of $D$ and denote by $S'_P$ the set of vertices of $S'$ belonging to $V_P$. Suppose for a contradiction that $S'_P \not\subseteq \{v_1, v_p, v_l\}$. If $|S'_P| = 1$, then $S'_P = \{v_p\}$ since $v_p$ is an extremal vertex, a contradiction. If $|S'_P| = 2$, consider the vertex $v_k$ in $S'_P$ different from $v_p$. Suppose without loss of generality that $k\leq p$. Since $S'$ is an \mgs, vertices $v_{p+1}, \dots, v_l$ are covered by a shortest path for which $v_p$ is an endpoint. We argue that $S = (S' \setminus\{v_k\})\cup \{v_1\}$ is also an \mgs. First, vertices $v_1,\dots,v_{p-1}$ are covered by the shortest path between $v_1$ and $v_p$. Any shortest path with one endpoint being $v_k$ can be either extended or shortened into a shortest path with $v_1$ replacing $v_k$ as an endpoint, hence $S$ is a geodetic set of $D$. Furthermore, $S$ has the same cardinality as $S'$, thus it is minimum.
    Finally, if $|S'_P| > 2$, it can be verified that $S = (S'\setminus S'_P)\cup \{v_1, v_p, v_l\}$ is a geodetic set of size at most $|S'|$ using similar arguments.
	\end{proof}
	
	\begin{myclaim}\label{claim-dipaths}
	    Let $P$ be a core dipath of some digraph $D$. Denote by $V_P$ the inner vertices of $P$ and number them in the order induced by the arcs of the dipath so that $V_P = \{v_1, \dots, v_l\}$. There exists an \mgs $S$ of $D$ such that $V_P\cap S\subseteq\{v_1\}$. 
	\end{myclaim}
	
	\begin{proof} 	Recall that a core dipath is an oriented core path whose arcs are all oriented in the same direction. 
	    Let $S'$ be an \mgs of $D$. Then, $v_l$ is covered by some shortest path ending at a vertex $w\in S'$ (possibly itself). If $S'\cap V_P = \emptyset$, then $S'$ satisfies the required properties and we are done. Suppose then that $|S'\cap V_P| \geq 1$. Define $S = (S'\setminus V_P)\cup\{v_1\}$. The unique path between $v_1$ and $v_l$ is by definition a shortest path, and thus $V_P\subseteq I(v_1, w)$. Furthermore, any shortest path with an endpoint in $S'\cap V_P$ can be replaced by a shortest path where $v_1$ replaces the vertex of $S'\cap V_P$, while covering the same vertices outside of $P$. Thus, $S$ is a geodetic set of $D$, and $|S|\leq|S'|$ so $S$ is also an \mgs of $D$. 
	\end{proof}
	
	\Cref{claim-path-3extremals,claim-path-2extremals,claim-path-1extremal,claim-dipaths} imply that at most two non-extremal vertices can be part of an \mgs in any oriented core path of $D$. Define $V_C$ as the set of core vertices of $D$ and $V_I$ as the set of inner vertices of oriented core paths that can be chosen in \Cref{claim-path-2extremals,claim-path-1extremal,claim-dipaths}.
	
	\begin{lemma}\label{lemma-fen-gs}
    Let $D$ be a digraph, with $S^0$ its set of extremal vertices. There exists an \mgs $S$ of $D$ such that $S\subseteq V_C\cup V_I\cup S^0$, and $|S\setminus S^0|\leq 8\fen(D)-8$.
	\end{lemma}
	
	\begin{proof}
	 By \Cref{claim-path-3extremals,claim-path-2extremals,claim-path-1extremal,claim-dipaths}, we can suppose there exists a \mgs $S'$ of $D$ whose vertices belonging to the base digraph of $D$ are in $V_C\cup V_I$. We can then apply \Cref{claim-trees-covered} on $S'$ to obtain the \mgs $S = (S'\cap V(D_b))\cup S^0$ where $D_b$ is the base digraph of $D$. Note that applying \Cref{claim-trees-covered} can only decrease the number of vertices in the obtained geodetic set. It follows that $S\subseteq V_C\cup V_I\cup S^0$. Let $P^c$ the set of core paths of $D_b$. Again by \Cref{claim-path-3extremals,claim-path-2extremals,claim-path-1extremal,claim-dipaths}, we have $|V_C\cup V_I| \leq |V_C| + 2|P^c|$, and by \Cref{obs-fen-bounds}, we obtain $|S\setminus S^0|\leq 8\fen(D)-8$.
	\end{proof}
	
	We can now describe the algorithm solving \gsfull parameterized by the feedback edge set number of the underlying undirected graph of the input digraph. This algorithm guesses the vertices to add in a solution among vertices in $V_C\cup V_I$. The algorithm can be stated as follows:
	
	\begin{itemize}
	    \item For all subsets $S^1$ of $V_C\cup V_I$, if $S=S^0\cup S^1$ is a geodetic set of $D$, mark $S$
	    \item Return the marked set $S$ of minimum size.
	\end{itemize}
	
	\begin{proof}[Proof of \Cref{thm-fenalgo}]
	The correctness of the algorithm is clear from \Cref{lemma-fen-gs}. The running time follows, since by \Cref{lemma-fen-gs}, we have $|V_C\cup V_I|\leq 8\fen(D)-8$. Thus, the algorithm checks $2^{\mathcal O(\fen(D))}$ different vertex sets. Each check is polynomial-time in $n$ since \gsfull belongs to \NP.
	\end{proof}

    \section{\texorpdfstring{\NP}{NP}-hardness on restricted DAGs}
    \label{sec:hard}
    
    In this section, we sketch the proof of Theorem~\ref{thm:fvn-hardness}. 
    We present a reduction from the classic \NP-complete problem \TDM~\cite{DBLP:books/fm/GareyJ79}.

    \problemdec{\TDM}{A ground set $U$ partitioned in three sets $X^\alpha$, $X^\beta$ and $X^\gamma$ such that $|X^\alpha| = |X^\beta| = |X^\gamma| = n$ and a collection of 3D edges $E\subseteq X^\alpha\times X^\beta\times X^\gamma$.}{Decide if there exists a set $S$ of $n$ edges of $E$ so that any element of $U$ is covered by an edge of $S$. The set $S$ is called a 3D-matching of $(U, E)$.}

    \paragraph{Reduction}
    Consider an instance of \TDM with  a ground set $U$ and its partition into three sets $X^\alpha$, $X^\beta$ and $X^\gamma$, and an edge set $E$. 
    In the following, the notation $\delta$ will designate any of the letters $\alpha$, $\beta$ or $\gamma$. Number the vertices of $X^\delta$ so that $X^\delta = \{x_1^\delta, \dots, x_n^\delta\}$ and denote by $m$ the cardinality of $E$. We construct a digraph $D$ as follows:
        
        \begin{description}
            \item[Edge vertices.] Add $n$ sets $M_1, \dots, M_n$ of $m$ vertices, each vertex corresponding to some edge of $E$. Denote by $u_i^e$ the vertex in $M_i$ associated with edge $e$. Add also for $1\leq i\leq n$ a vertex $d_i$ and an arc from any vertex in $M_i$ to $d_i$.
            \item[Ensuring edge vertices are covered.] Add three vertices $a$, $b$ and $c$ and connect them as follows:
            \begin{itemize}
                \item Add arcs from $a$ to each vertex in $\bigcup_{1\leq i\leq n}M_i$.
                \item Add arcs from each vertex in $\bigcup_{1\leq i\leq n}M_i$ to $b$.
                \item Add arcs from each vertex in $\{d_i\mid 1\leq i\leq n\}$ to $c$.
                \item Add an arc from $a$ to $c$.
            \end{itemize}
            \item[Element vertices.] Add vertices $v_i^\delta$, $w_i^\delta$ and $t_i^\delta$ associated with the element $x_i^\delta$ in $X^\delta$. Add arcs from $w_i^\delta$ to $v_i^\delta$ and from $t_i^\delta$ to $v_i^\delta$.
            \item[Encoding adjacency.] Add nine vertices $\alpha_1$, $\alpha_2$, $\alpha_3$, $\beta_1$, $\beta_2$, $\beta_3$, $\gamma_1$, $\gamma_2$, $\gamma_3$ and add an outgoing pendant vertex to each of them. For $i\in\{1,2,3\}$, denote by $\delta'_i$ the pendant vertex associated with $\delta_i$. Denote by $e$ some edge of $E$ such that $e=(x_i^\alpha, x_j^\beta, x_k^\gamma)$ and call $u_e$ any of its associated vertices in $\bigcup_{1\leq i\leq n}M_i$. Define $\lambda = n+1$. This value is used to define some additional paths, that will encode the adjacency of 3D edges with respect to elements of $U$. More precisely, we add the following paths:
            \begin{itemize}
                \item Add a path of length $\lambda^2-i\lambda$ from $u_e$ to $\alpha_1$, a path of length $\lambda^2$ from $u_e$ to $\alpha_2$, and a path of length $\lambda^2+i\lambda$ from $u_e$ to $\alpha_3$;
                \item Add a path of length $\lambda^2-j\lambda$ from $u_e$ to $\beta_1$, a path of length $\lambda^2$ from $u_e$ to $\beta_2$, and a path of length $\lambda^2+j\lambda$ from $u_e$ to $\beta_3$
                \item Add a path of length $\lambda^2-k\lambda$ from $u_e$ to $\gamma_1$, a path of length $\lambda^2$ from $u_e$ to $\gamma_2$, and a path of length $\lambda^2+k\lambda$ from $u_e$ to $\gamma_3$;
                \item For any $\delta\in\{\alpha,\beta,\gamma\}$, add a path of length $\lambda^2+l\lambda$ from $\delta_1$ to $w_l^\delta$, a path of length $\lambda^2$ from $\delta_2$ to $w_l^\delta$, a path of length $\lambda^2$ from $\delta_2$ to $t_l^\delta$ and a path of length $\lambda^2-l\lambda$ from $\delta_3$ to $w_l^\delta$.
            \end{itemize}
            \item[Ensuring the edge paths are covered.] For each path from an edge vertex of $\bigcup_{1\leq i\leq n}M_i$ to a vertex $\delta_i$ for $i\in\{1,2,3\}$ a pendant outgoing vertex to the vertex right before $\delta_i$ in said path.
            \item[Shortcutting the vertex $\textbf{a}$.] Add all arcs from $a$ to vertices $v_i^\delta$ associated with elements of $U$.
        \end{description}
        
        \Cref{fig-whole-reduc,fig-detail-reduc} 
        illustrate the main components of the obtained graph, namely the edge vertices, the element vertices and the encoding of adjacencies described above. Note that the choice of $\lambda$ ensures that the construction is valid. In particular, any path added during the \enquote{Encoding adjacency} step described above has positive length, since $\lambda^2 - i\lambda \geq n$ for any $i\leq n$.
        We argue that there exists a solution to \TDM for the original instance if and only if there exists a geodetic set of size $9nm+4n+12$ in $D$.
    \newcounter{iSet}

\NewDocumentCommand{\drawISet}{m m m O{}}{%
    \stepcounter{iSet}
    \begin{scope}[every node/.style={draw, circle}, #4]
        \node (n\theiSet_1) at ($(0,1) + (#1,#2)$) {};
        \node (n\theiSet_2) at ($(0,0) + (#1,#2)$) {};
        \node (n\theiSet_3) at ($(0,-1) + (#1,#2)$) {};
        
        \node[shape=ellipse, fit=(n\theiSet_1) (n\theiSet_2) (n\theiSet_3)] (s\theiSet) [label=below:#3]{};
    \end{scope}
}

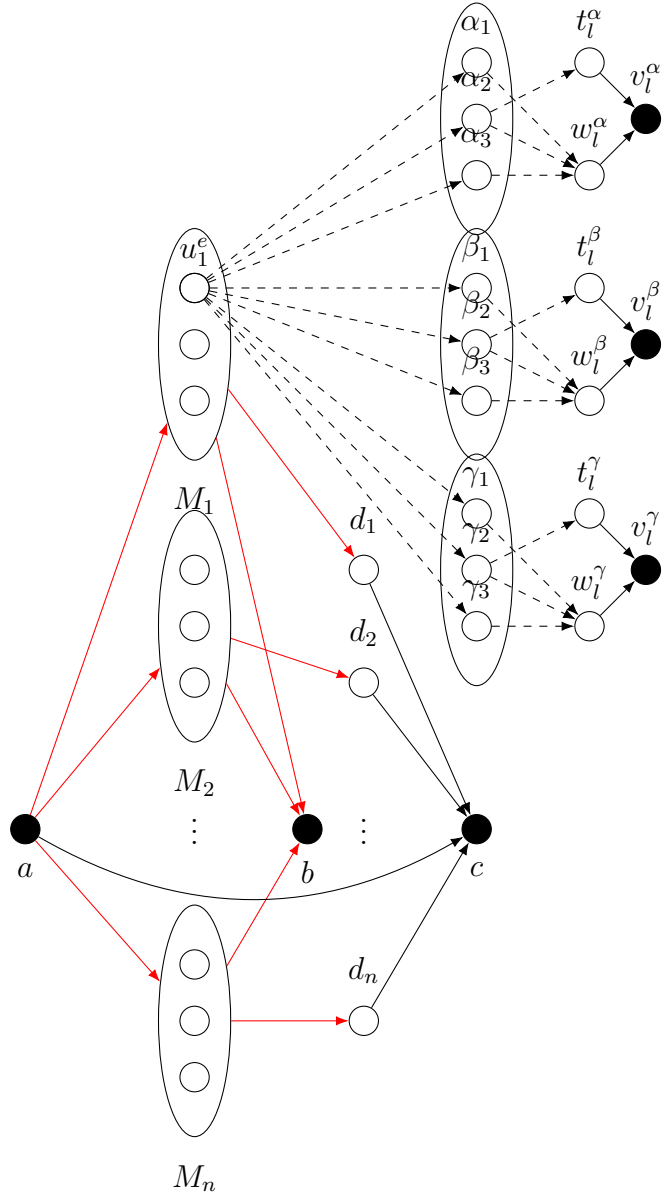
\begin{figure}
    \centering
    \begin{tikzpicture}[scale=.75, >=Latex]
    
        \drawISet{0}{6}{$M_1$}
        \drawISet{0}{1}{$M_2$}
        \drawISet{0}{-6}{$M_n$}
        \node at (0,-2.5) {\vdots};
        \node at (3,-2.5) {\vdots};
        
        \begin{scope}[every node/.style={draw, circle}]
        \node[label=below:$a$, fill] (a) at (-3, -2.6) {};
        \node[label=below:$b$, fill] (b) at ( 2, -2.6) {};
        \node[label=below:$c$, fill] (c) at ( 5, -2.6) {};
        \foreach \i/\n/\y in {1/1/2,2/2/0,3/n/-6} {
            \node[label=$d_\n$] (d\i) at (3,\y) {}; 
            
        };
        \foreach \i in {1,2,3} {
            \begin{scope}[every path/.style={draw, ->, red}]
                \path (a) -- (s\i);
                \path (s\i) -- (d\i);
                \path (s\i) -- (b);
                \path[black] (d\i) -- (c);
            \end{scope}
        };
        
        \draw[->] (a) to[bend right] (c); 
        
        \node[label={[label distance=-.2cm]:$u_1^e$}] (u) at (n1_1) {};
        \end{scope}
        
        \foreach \i/\y/\letter in {alpha/10/\alpha, beta/6/\beta, gamma/2/\gamma} {
            \begin{scope}[every node/.style={draw, circle}, label distance=-0.15cm]
                \node[label=$\letter_1$] (\i_1) at ($(0,1) + (5,\y)$) {};
                \node[label=$\letter_2$] (\i_2) at ($(0,0) + (5,\y)$) {};
                \node[label=$\letter_3$] (\i_3) at ($(0,-1) + (5,\y)$) {};
                
                \node[shape=ellipse, fit=(\i_1) (\i_2) (\i_3)] (s\i) {};
                \node[label=$t_l^\letter$] (t_\i) at ($(0, 1) + (7, \y)$) {};
                \node[label=$w_l^\letter$] (w_\i) at ($(0,-1) + (7, \y)$) {};
                \node[label=$v_l^\letter$, fill] (v_\i) at ($(1, 0) + (7, \y)$) {};
            \end{scope}
            \begin{scope}[every path/.style={draw, ->}]
                \foreach \j in {1, 2, 3} {
                    \path[dashed] (\i_\j) -- (w_\i);
                };
            \path[dashed] (\i_2) -- (t_\i);
            \path (t_\i) -- (v_\i);
            \path (w_\i) -- (v_\i);
            \end{scope}
        };
        
        \foreach \i in {alpha,beta,gamma} {
            \foreach \j in {1,2,3} {
                \begin{scope}[every path/.style={draw, ->, dashed}]
                    \path (u) to (\i_\j);
                \end{scope}
            };
        };
    \end{tikzpicture}
    \caption{A partial representation of the digraph $D$ constructed during the reduction. Dashed arcs represent paths of length greater than one. Red arcs between a vertex and a vertex set mean that there exists such an arc for all vertices in the vertex set. Filled vertices are extremal vertices of $D$, that belong to any geodetic set. Arcs adjacent to vertices $\delta_i$ are represented for only one edge vertex and vertices associated with three different elements, each of them belonging to a different partition set of $U$. Pending vertices to $\delta_i$ and paths from edge vertices to $\delta_i$ are not represented.
    }
    \label{fig-whole-reduc}
\end{figure}

\begin{lemma}
            \label{lemma-fvs-hardness-forward}
            If \((X^\alpha, X^\beta, X^\gamma, E)\) is a \yes-instance of \TDM,
            then \(D, k\) is a \yes-instance of \gsfull.            
        \end{lemma}
  \begin{proof}
        Suppose that there exists a set $S$ of $n$ edges that cover all elements of $U$. We define a set $S'$ of vertices of $D$ of size $9nm+4n+12$ and prove that it is a geodetic set in $D$. First, we add all extremal vertices (\emph{i.e.} vertices that have only incoming or outgoing arcs) since they belong to any geodetic set. In $D$, they correspond to vertices $a$, $b$, $c$, pendant vertices to $\delta_i$, pendant vertices to paths from edge vertices to $\delta_i$ for $i\in\{1,2,3\}$ and all element vertices $v_l^\delta$, which amounts to $9nm+3n+12$ vertices. Number the edges of $S$ so that $S = \{e_1, \dots, e_n\}$. For each edge $e_i$, we add the vertex of $M_i$ associated with edge $e_i$ in $S'$. We verify now that all vertices in $D$ that do not belong to $S'$ are located on a shortest path between two vertices of $S'$. First, notice that vertices $a$ and $b$ cover all edge vertices in $\bigcup_{1\leq i\leq n}M_i$. All inner vertices of paths between edge vertices of $\bigcup_{1\leq i\leq n}M_i$ and vertices $\delta_i$ for $i\in\{1,2,3\}$ are covered by $a$ and the pendent vertices associated to each of those paths. Vertices $\delta_i$ for $i\in\{1,2,3\}$ are covered by $a$ and the pendent vertices associated to vertices $\delta_i$. For $1\leq i\leq n$, because one vertex $u$ of $M_i$ is selected in $S'$, the vertex $d_i$ is covered by the path between $u$ and $c$. The only vertices not yet considered are vertices on paths between $\delta_i$ (excluded) and $w_l^\delta$ (included) for $1\leq l\leq n$, and vertices between $\delta_2$ (excluded) and $t_l^\delta$ (included). Consider one element $x^\delta_j$. Because $S$ is a solution of \TDM for the original instance, there exists an edge $e_i$ of $S$ covering $x_l^\delta$. By definition of $D$, this means that:
        \begin{itemize}
            \item the path between $u_i^{e_i}$ and $\delta_1$ has length $\lambda^2-l\lambda$;
            \item the path between $u_i^{e_i}$ and $\delta_3$ has length $\lambda^2+l\lambda$;
            \item the path between $\delta_1$ and $w_j^\delta$ has length $\lambda^2+l\lambda$;
            \item the path between $\delta_3$ and $v_l^\delta$ has length $\lambda^2-l\lambda$.
        \end{itemize}
        
        This means in particular that the paths from $u_i^{e_i}$ to $w_l^\delta$ respectively going through $\delta_1$, $\delta_2$ and $\delta_3$ are all of length $2\lambda^2$. Furthermore, the path between $u_i^{e_i}$ and $t_l^\delta$ going through $\delta_2$ is also of length $2\lambda^2$. This means that all paths between $u_i^{e_i}$ and $v_l^\delta$ have length $2\lambda^2+1$ and all vertices belonging to those paths are covered. This proves that $S'$ is a geodetic set of $D$.
        \end{proof}

            \begin{figure}
        \centering
        \begin{tikzpicture}[>=Latex, scale=.9]
            \begin{scope}[every node/.style={draw, circle}]
                \node[label=$u^e$] (u) at (-2,0) {};
            \end{scope}
            
            \foreach \i/\y/\letter/\k in {alpha/6.5/\alpha/l_1, beta/0/\beta/l_2, gamma/-6.5/\gamma/l_3} {
                \begin{scope}[every node/.style={draw, circle}, label distance=-0.15cm]
                    \node[label=$\letter_1$] (\i_1)   at ($(0,1.5) + (5,\y)$) {};
                    \node (p_\i_1) at ($(-1.5,1.5) + (5,\y)$) {};
                    \node[label=$\letter_2$] (\i_2) at ($(0,0) + (5,\y)$) {};
                    \node (p_\i_2) at ($(-1.5,0) + (5,\y)$) {};
                    \node[label=$\letter_3$] (\i_3) at ($(0,-1.5) + (5,\y)$) {};
                    \node (p_\i_3) at ($(-1.5,-1.5) + (5,\y)$) {};
                    
                    \node[shape=ellipse, fit=(\i_1) (\i_2) (\i_3)] (s\i) {};
                    \node[label=$t_{\k}^\letter$] (t_\i) at ($(0, 1.5) + (9, \y)$) {};
                    \node[label=below:$w_{\k}^\letter$] (w_\i) at ($(0,-1.5) + (9, \y)$) {};
                    \node[label=$v_{\k}^\letter$, fill] (v_\i) at ($(1, 0) + (9, \y)$) {};
                \end{scope}
                \begin{scope}[every path/.style={draw, ->}]
                    \foreach \j in {1, 2, 3} {
                        \node[above right=.6cm of p_\i_\j, fill, circle] (pend_\i_\j) {};
                        \node[above right=.6cm of \i_\j, fill, circle] (p\i\j) {};
                        \path[] (\i_\j) -- (p\i\j);
                        \path (p_\i_\j) -- (\i_\j);
                        \path (p_\i_\j) -- (pend_\i_\j);
                    };
                \path[dashed] (\i_1) -- (w_\i) node [pos=0.6, above, sloped] {$\lambda^2+\k\lambda$};
                \path[dashed] (\i_2) -- (w_\i) node [pos=0.5, above, sloped] {$\lambda^2$};
                \path[dashed] (\i_3) -- (w_\i) node [pos=0.4, above, sloped] {$\lambda^2-\k\lambda$};
                \path[dashed] (\i_2) -- (t_\i) node [midway, above, sloped] {$\lambda^2$};
                \path (t_\i) -- (v_\i);
                \path (w_\i) -- (v_\i);
                \path[dashed] (u) -- (p_\i_2) node [pos=.8, above, sloped] {$\lambda^2-1$};
                \end{scope}
            };
            
            \foreach \i/\k in {alpha/i,beta/j,gamma/k} {
            \draw[dashed, ->] (u) -- (p_\i_1) node [midway, above, sloped] {$\lambda^2-\k\lambda-1$};
            \draw[dashed, ->] (u) -- (p_\i_3) node [midway, below, sloped] {$\lambda^2+\k\lambda-1$};
                
            };
        \end{tikzpicture}
        \caption{Details of the gadget ensuring adjacency of edges and elements. One edge vertex $u^e$ is represented, as well as three gadgets that each correspond to an element of a set $X^\delta$. Dashed arcs represent paths, whose length is indicated on the drawing. Filled vertices are extremal vertices that belong to any geodetic set of the digraph.}
        \label{fig-detail-reduc}
    \end{figure}
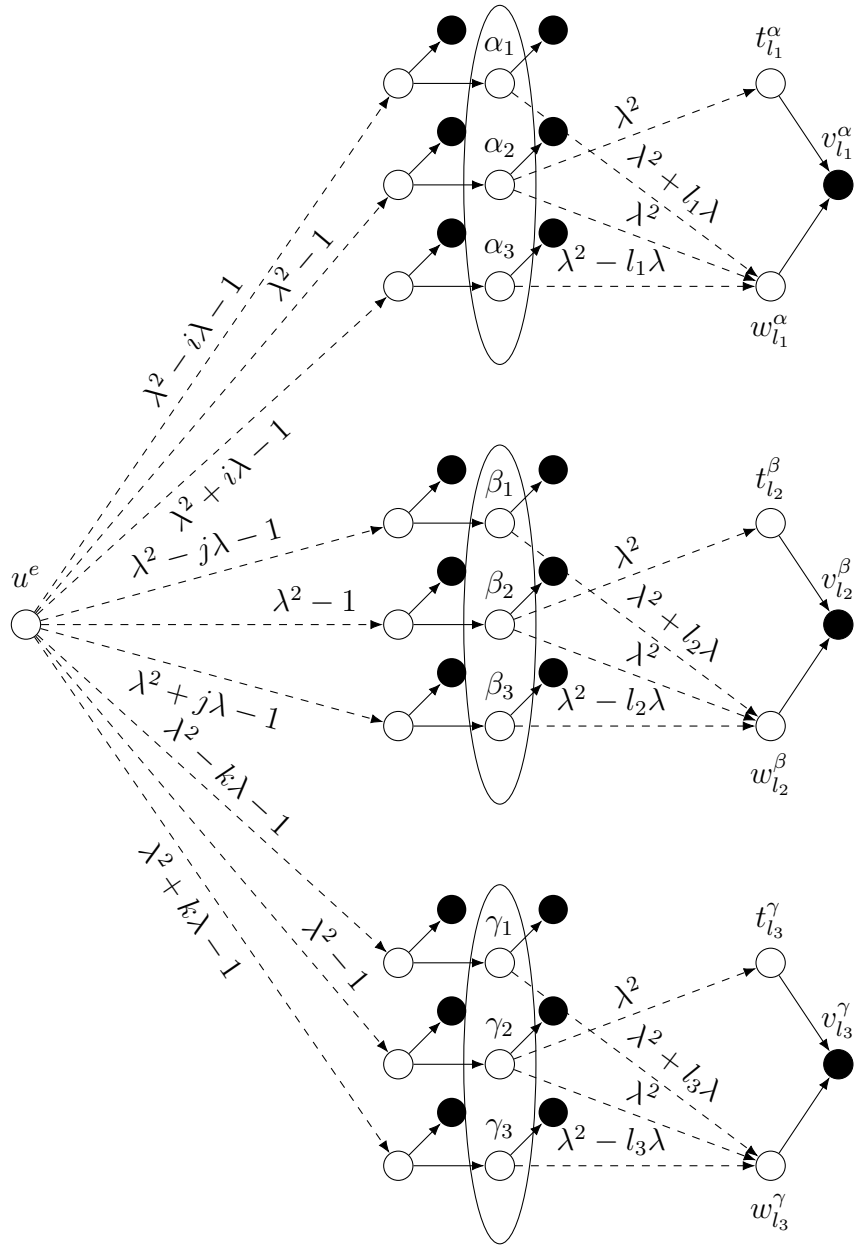
    
        \begin{lemma}
        \label{lemma-fvs-hardness-backward}
            If $(D, k)$ is a \yes-instance of \gsfull then
            $(X^\alpha, X^\beta, X^\gamma, E)$ is a \yes-instance of \TDM.
        \end{lemma}
        \begin{proof}
        Suppose that there exists a geodetic set $S'$ of size $9nm+4n+12$ in $D$. We define then a set $S$ of edges of $E$ and argue this is a 3D matching of the original instance. As said previously, all extremal vertices of $D$ belong to $S'$. This means that there are $n$ vertices of $S'$ that are yet to be identified. Consider vertex $d_i$ for $1\leq i\leq n$ and suppose it is not in $S'$. The only reachable vertex from $d_i$ is $c$, thus $d_i$ must be covered by a shortest path ending in $c$ and starting at some vertex that can reach $d_i$. Such vertices are exactly vertices of $M_i\cup \{a\}$, but since the arc $(a,c)$ is present in $D$, one vertex of $M_i$ must belong to $S'$. This means that for all $1\leq i\leq n$, one vertex of $M_i\cup\{d_i\}$ belongs to $S'$ and because of the cardinality of $S'$, there is exactly one vertex of $M_i\cup\{d_i\}$ in $S'$. One can further suppose that this vertex is in $M_i$, since all vertices that $d_i$ could cover (\emph{i.e.} vertices that can reach $d_i$ or are reachable by it) are already covered by extremal vertices or belong to $S'$ as seen previously. Each vertex of $S' \cap \bigcup_{1\leq i\leq n}M_i$ is associated to one edge of $E$. Define $S^*$ as the set of all such edges. By definition, $|S^*|\leq n$. If the inequality is strict, add arbitrarily $n-|S^*|$ edges to $S^*$ and call $S$ this augmented edge set of cardinality $n$. Consider $x_l^\delta$ an element of $U$ and $t_l^\delta$ its associated vertex in $D$. Since $S'$ is a geodetic set, and $t_l^\delta$ does not belong to $S'$, it is located on a shortest path from a vertex of $S'$ and $v_l^\delta$ (the only reachable vertex from $t_l^\delta$). Among the extremal vertices of $D$, only $a$ can reach $t_l^\delta$, and since the arc $(a, v_l^\delta)$ is present in $D$, no path between $a$ and $v_l^\delta$ going through $t_l^\delta$ is a shortest path. This means that $t_l^\delta$ must be covered by a path between the edge vertex $u^e$ of $S'$ and $v_l^\delta$. Without loss of generality, suppose that $\delta=\alpha$ (the reasoning is similar for $\beta$ and $\gamma$) and suppose that the edge $e$ corresponding to vertex $u^e$ contains the element $x_i^\alpha$ of $X^\alpha$. Note that $e\in S^*$. By definition of $D$, four different paths between $u^e$ and $v_l^\alpha$ are present in $D$:
        \begin{itemize}
            \item One going through $\alpha_1$ of length $2\lambda^2+\lambda(l-i)+1$;
            \item One going through $\alpha_3$ of length $2\lambda^2 + \lambda(i-l)+1$;
            \item One going through $\alpha_2$ and $w_l^\alpha$ and one going through $\alpha_2$ and $t_l^\alpha$, both of length $2\lambda^2+1$.
        \end{itemize}
        In order for $t_l^\alpha$ to be covered, the path between $u^e$ and $v_l^\alpha$ going through it has to be a shortest path. This means that $2\lambda^2+1\leq \lambda^2+\lambda(l-i)+1$ and $2\lambda^2+1\leq \lambda^2+\lambda(i-l)+1$, which ensures that $i=l$. In turn, this means that $e$ covers element $x_l^\alpha$. Since this reasoning can be applied for any element of $U$, we proved that $S^*$ is a 3D matching of $U$, and that $S$ is a 3D matching of $U$ of size $n$.
        \end{proof}

\begin{proof}[Proof of Theorem~\ref{thm:fvn-hardness}]
First, note that \gsfull is known to be \NP-complete on general digraphs, thus it belongs to \NP. 
Also, note that deleting vertices $\alpha_1$, $\alpha_2$, $\alpha_3$, $\beta_1$, $\beta_2$, $\beta_3$, $\gamma_1$, $\gamma_2$, $\gamma_3$, $a$, $b$ and $c$ 
removes any cycle in the underlying graph of $D$, thus its feedback vertex number is upper-bounded by $12$. Furthermore, removing those vertices yields a disjoint union of subdivided stars as the remaining graph. It is not difficult to verify that the pathwidth of a disjoint union of graphs is the maximum of the pathwidth of its components, and that the pathwidth of any subdivided star is 2 (each bag contains the central vertex of the star; each other bag contains the two endpoints of one of the edges, and the bags are ordered naturally following the orders of the paths representing the branches of the the star). Consider then a path decomposition $\mathcal D$ of width 2 of this disjoint union of subdivided stars, and create a path decomposition of the original graph by adding the twelve above vertices in all the bags of the decomposition $\mathcal D$. This is a valid path decomposition of width 14 for the underlying graph of $D$, so the pathwidth of the underlying graph of $D$ is at most 14.
The digraph $D$ is acyclic and the reduction can be completed in time polynomial in the input instance.
The correctness of the reduction follows from \Cref{lemma-fvs-hardness-forward,lemma-fvs-hardness-backward}.
This, along with the fact that \TDM is \NP-complete~\cite{DBLP:books/fm/GareyJ79}, concludes the proof of Theorem~\ref{thm:fvn-hardness}.
\end{proof}

\section{Conclusion}\label{sec:conclu}
In this article, we continued the study of \gsfull for digraphs.
As directions for further research, it would be interesting to 
extend our algorithm for ditrees to more general classes of 
digraphs. For instance, what can be said about directed cactus 
graphs (dicactii)? 
An algorithm for directed cactus graphs without 
\(2\)-cycles (oriented cactii) was given in~\cite{AraujoA22}.
More broadly, one could consider directed outerplanar graphs.
Note that 
an algorithm for undirected outerplanar graphs appears 
in~\cite{mezzini2018}.  
Regarding parameterized complexity, it would be natural to ask 
whether a polynomial kernel exists with respect to the feedback 
edge set number of the underlying undirected graph. This remains 
open even for undirected graphs~\cite{KK22}. 

%
%
%
\bibliographystyle{splncs04}
\bibliography{references}

\end{document}